\documentclass[preprint,12pt]{elsarticle}

\usepackage{amsmath,amsfonts,amsthm,amssymb,paralist,subfigure,graphicx,amsbsy,float,epsfig,color,xcolor}
\usepackage{cuted,mathtools,lipsum}

\usepackage{soul}
\usepackage{algorithm,algpseudocode}
\usepackage{stfloats}
\usepackage[toc]{appendix}
\usepackage{blkarray}
\usepackage{bbm}
\usepackage{bm}
\usepackage[makeroom]{cancel}
\usepackage{cases}
\usepackage{changebar}
\usepackage{epsfig}
\usepackage{graphicx}
\usepackage{flushend}
\usepackage{textcomp}
\usepackage{xcolor}
\usepackage{flushend}
\usepackage{graphicx}
\usepackage{mathtools}
\usepackage{mathrsfs}
\usepackage{setspace}
\usepackage{soul}
\usepackage{subfigure}
\usepackage{stfloats}
\usepackage{tikz}
\usepackage{xspace}
\usepackage{cuted,mathtools,lipsum}
\usepackage{pgfplots}

\usepackage{stmaryrd,mathrsfs,url}
\usepackage{pifont}
\usepackage{hyperref}       
\hypersetup{
	colorlinks=true,
	linkcolor=cyan,
	filecolor=mnodea,
	urlcolor=cyan,
	citecolor=lime,
}

\usepackage{url}            
\usepackage{booktabs}       
\usepackage{nicefrac}       
\usepackage{lipsum}
\usepackage{dsfont}

\newtheorem{thm}{Theorem}
\newtheorem{lem}{Lemma}
\newtheorem{ass}{Assumption}
\newtheorem{definition}{Definition}
\newtheorem{remark}{Remark}

\def\mb{\mathbf}

\def\mc{\mathcal}

\journal{Systems \& Control Letters}

\begin{document}

\begin{frontmatter}
\title{Resilient Monitoring of Social Dynamical Systems through Collaborative Multi-Agent Networks under Latency
}	

\author[MD]{Mohammadreza Doostmohammadian}\ead{doost@semnan.ac.ir},
\author[SP1,SP2]{Sergio Pequito}\ead{sergio.pequito@tecnico.ulisboa.pt}

\address[MD]{Mechatronics Department, Faculty of Mechanical Engineering, Semnan University, Iran}

\address[SP1]{Dept. of Electrical and Computer Engineering, Instituto Superior T\'ecnico, University of Lisbon, Portugal}

\address[SP2]{Institute for Systems and Robotics, Instituto Superior T\'ecnico, University of Lisbon, Portugal}

\begin{abstract}
	Social dynamical networks significantly influence contemporary digital landscapes, affecting realms from social activism to public policy formulation. This paper investigates the use of multi-agent systems (MAS) to monitor and analyze these networks. Firstly, we propose a single-time-scale distributed inference model designed to effectively manage challenges such as latency and agent failure. Secondly, we provide sufficient conditions that ensure the stability of the proposed scheme. Notably, the observer gain design remains effective regardless of time delays. Thirdly, we develop a computationally efficient recovery mechanism for agent failures that relies on employing graph-theoretic approaches to restore network observability by replacing failed agents (due to sensing failure or unbounded delays, i.e., packet drops) by implementing computationally efficient graph-theoretic methods to assign observationally equivalent agent counterparts. Lastly, we illustrate the proposed scheme through a pedagogical example and real-world network applications.
\end{abstract}

\begin{graphicalabstract}
	\includegraphics{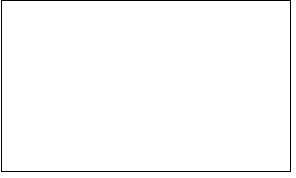}
\end{graphicalabstract}

\begin{highlights}
	\item Clearly defining the formal mathematical framework for monitoring of social dynamical systems 
	\item Addressing distributed social inference under latency
	\item Developing graph-theoretic recovery mechanisms for agent failures 
\end{highlights}


\end{frontmatter}

\section{Introduction} \label{sec_intro}
Social dynamical networks are increasingly integral to our digital world \cite{proskurnikov2017tutorial}. The evolution of these networks can profoundly impact social activism, potentially disrupting core societal functions. This can occur proactively or in response to natural disasters like hurricanes or earthquakes. Therefore, monitoring these networks to develop strategies that can alleviate negative social impacts is crucial. Additionally, these networks can facilitate the provision of essential services to a broader audience, help formulate public policies that reflect the population's concerns, and enhance digital marketing efforts. 

A possible solution may rely on multi-agent systems (MAS) to effectively monitor the dynamics of social networks. In these systems, agents collect localized information about the network and share it with other agents. This collaborative process enables the agents to infer the overall state of the social network, ensuring a comprehensive monitoring platform. The setting enabled by MAS promotes decentralized methods for learning over social networks \cite{kayaalp2024social,Cirillo2023Memory,salhab2020social,rahimian2016group}, which contrasts with a centralized scheme where information is routed towards a central fusion station for inference. In decentralized methods, inference can be achieved by establishing consensus among various snapshots of states measured by different agents. This process significantly increases communication between agents in the MAS and can be further complicated by the size of the states being exchanged \cite{he2019secure,battilotti2021stability,mo2020distributed,qian2022consensus,ryu2023consensus,spl24}.

Alternatively, there is the option of \emph{single-time-scale distributed inference} (i.e., without an inner consensus loop between snapshots and with only a single step of averaging/consensus). This includes diffusion-based strategies \cite{chen2012diffusion,bordignon2022partial}, consensus-based observers \cite{7463019}, Byzantine-resilient methods \cite{mitra2019byzantine}, and approaches resilient to sensor attacks \cite{chen2018resilient}, which help mitigate link failures \cite{kar2009distributed} and are robust against limited communications \cite{khan2014collaborative}. Additional research focuses on distributed estimation schemes that are robust under time-varying network topologies \cite{mitra2021distributed,paritosh2022distributed} and consensus in the presence of delays \cite{Themis_delay}. For a review of state-of-the-art distributed estimation methods, refer to \cite{he2020distributed}.

Unfortunately, real-world applications of MAS often encounter two significant challenges. Firstly, an agent may fail due to unbounded delays and packet loss, despite extensive efforts to design robust and resilient networks--a recent survey and its references highlight these issues \cite{pirani2023graph}. Secondly, delayed data exchanges between agents, often called latency, are common in networked control systems but have been largely overlooked in the existing literature on distributed algorithms \cite{chen2012diffusion,bordignon2022partial,7463019,mitra2019byzantine,khan2014collaborative,mitra2021distributed,paritosh2022distributed,pirani2023graph,asilomar11}. Particularly, the design of delay-tolerant collaborative algorithms along with addressing agent failure recovery advances the novelty as compared to more closely related works (e.g., \cite{7463019,khan2014collaborative,pirani2023graph}). Although the works \cite{7463019,khan2014collaborative,pirani2023graph}  address decentralized solutions for tracking and estimation, are vulnerable to possible time-delays in the data-sharing network of agents and provide no solution for observationally-equivalent recovery in case of agent failure.
To the best of the authors' knowledge, there remains a notable gap concerning decentralized schemes that effectively track potentially unstable social dynamics in the presence of bounded heterogeneous time delays in the possible presence of faulty agents.

In this paper, we propose a single-time-scale distributed inference scheme for potentially unstable social dynamical networks that may be prone to sensing loss and latency across a multi-agent network due to faults, congestion, or environmental conditions. The main contributions of our work are threefold: 

\emph{(i)} We introduce a single-time-scale distributed inference protocol based on the Kronecker composite network product and distributed observability. This protocol combines consensus on a-priori estimates with an observation update that accommodates potentially heterogeneous bounded time delays. 

\emph{(ii)} We establish sufficient conditions for the stability of this scheme. Notably, we demonstrate that the design of the observer gain remains effective irrespective of the time delays.

\emph{(iii)} In the event of agent failures, we propose a recovery mechanism using computationally efficient graph-theoretic methods. In a distributed sense, this strategy restores observability by substituting the failed agents with observationally equivalent ones. 

\textit{Paper Organization:} The rest of the paper is organized as follows. Section~\ref{sec_probStat} formulates the problem. Section~\ref{sec_mainResults} presents the main results on resilient-to-delay distributed observer design. Section~\ref{sec_sim} provides the simulations and Section~\ref{sec_con} concludes the paper.

\textit{Notations:} Capital letters are used to represent matrices, and bold small letters represent vectors. $I_N$ and $0_N$ denote the identity and all-zero matrices of size $N$. The operator $\otimes$ denotes the Kronecker product. The operator ``$;$" denotes the column concatenation of vectors. The operator $\succ$ denotes positive-definiteness. Table~\ref{tab_notation} summarizes the notations in this paper:
\begin{table}
		\caption{Description of notations and symbols}
		\setlength{\tabcolsep}{0.7\tabcolsep}
		\centering \label{tab_notation}
		\begin{tabular}{ *{2}{c} }
			\hline
			\hline
			\textbf{Symbol} & \textbf{Description} \\
			\hline
			$\mb{x}_k,x_k^i$ & social state (at individual $i$)\\
			$k$ & time index \\
			$A$ & social dynamics matrix   \\ 
			$\mb{y}_k,\mb{y}_k^i$ & agent observation (at node $i$)  \\ 
			$H,H_i$ & observation matrix (at node $i$) \\
			$\nu_k,\zeta_k$ & system and output noise\\
			$Q,R$ & system and output noise covariance\\
			$n$ & number of social individuals  \\
			$N$ & number of agents  \\
			$\mathcal{G}_W$ & multi-agent network  \\
			$\mathcal{G}_A$ & social network\\
			$W$ & multi-agent weight matrix \\
			$\widehat{\mb{x}}_k$ & observed/estimated state\\
			$\tau_{ij}$ & time-delay at link $(j,i)$  \\
			$\overline{\tau}$ &  maximum time-delay  \\
			$K$ & feedback gain matrix \\
			$\widehat{A},\underline{\widehat{A}}$ &  closed-loop error matrix  \\
			$\mb{e}_k,\underline{\mb{e}}_k $ & estimation error \\
			$D_H,\overline{D}_H$ & block-diagonal observation matrix  \\
			$\rho(\cdot)$ & spectral radius operator    \\						
			\hline \hline
	\end{tabular}
\end{table}

\section{Problem Formulation} \label{sec_probStat}
\subsection{System Description}
	We consider a social dynamical network monitored by a multi-agent system with the following components: social dynamics, multi-agent observations, and communication network. These components are illustrated in Fig.~\ref{fig_multiagent} and described in detail below.
\begin{figure} 
	\centering
	\includegraphics[width=3.3in]{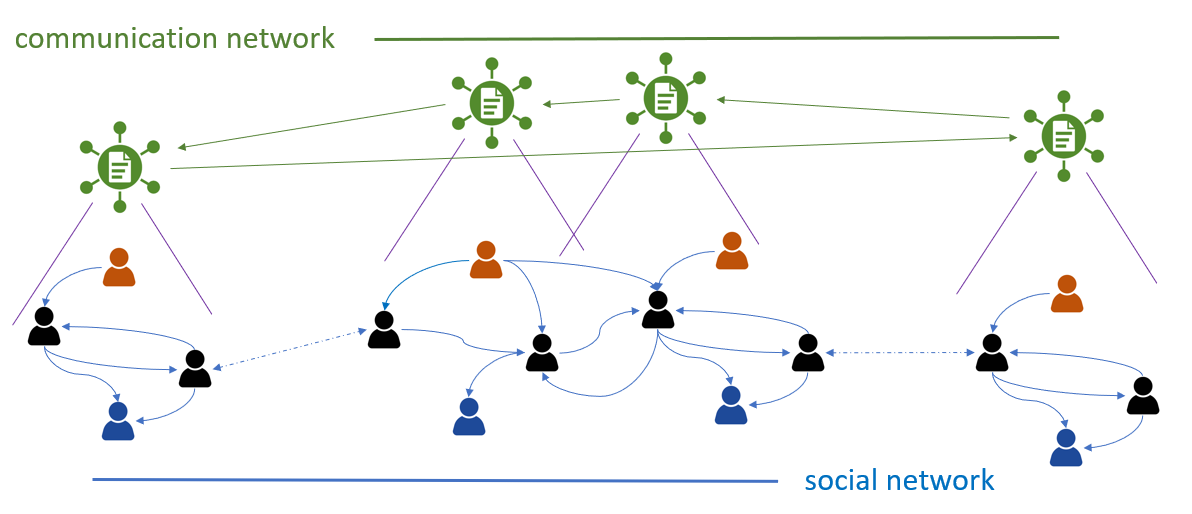}
	\caption{This figure illustrates a social network monitored by a multi-agent network. The social network has the interactions between individuals depicted by blue directed links, and where some individuals might be just influencers (in red) and others just followers (in blue). Additionally, agents depicted in green harvest data from select individuals enclosed by purple boundaries. These collectors then exchange gathered data through a dedicated communication network, illustrated by green lines.} \label{fig_multiagent}
\end{figure}  

\textbf{Social Dynamics:}
	This is characterized by the interaction and exchange of data between individuals over time. Each individual's state (e.g., an opinion or belief) is updated by their previous state and others' every time data is received. A possible model, often used to characterize such updates between $n$ individuals, is as follows~ \cite{khan2014collaborative,jia2015opinion}:
	\begin{align}\label{eq_sys1}
		\mb{x}_{k+1} = A\mb{x}_k + \nu_k,\qquad k\geq 0,
	\end{align}
	where the vector $\mb{x}_k \in \mathbb{R}^n$ encapsulates the \emph{social state} at time-step $k\in\mathbb{N}$  comprising the $i$th individuals' scalar states $x_k^{i}$, i.e., $\mb{x}_{k} = \left({x}_{k}^1; \ldots;{x}_{k}^n \right)$, and $\boldsymbol{\nu}_k \in \mathbb{R}^n$ represents the process noise with covariance $Q$, which may include possible rounding errors in the data storage and communication, modelled as standard Gaussian noise. Often in the literature, social interactions effectively are modelled using graph theory, where linear systems are prevalent. For example, if nodes (individuals) interact with weights (influence), linear relationships can capture the essence of how influence spreads through the social network. As an example, refer to the Freidkin and Johnsen model \cite{Friedkin0,Friedkin} or the French model \cite{French}. This linear framework serves well as a baseline model to identify key relationships and dynamics. Once established, it can be expanded to incorporate nonlinear dynamics if needed. 

The matrix~$A\in\mathbb{R}^n$ denotes the \emph{social system matrix} with the scalar entries $a_{ij}$ describing the weight that individual $i$ sets to incorporate the incoming state data (i.e., opinion, belief) of individual $j$. Therefore, this matrix describes the influence that an individual exerts on others -- either directly or indirectly through third parties. 
Hereafter, the only assumption we make on this matrix is as follows:
\begin{ass} \label{ass_fullRank}
	The social system matrix $A$ is \mbox{full-rank}.
\end{ass}
Notice that this assumption is mild as individuals update their state using their own weight $a_{ii}\neq 0$, and the remaining weights are arbitrarily chosen by the individuals; hence, $A$  is \textit{almost-surely} full-rank in the generic sense, see~\cite[Section~2.1]{RAMOS2022110229} and~\cite[Supplementary Information]{Liu-nature}\footnote{
			By continuity, introducing non-zero off-diagonal entries with independently selected values at random following a continuous distribution (e.g., normal distribution), the probability of obtaining a non-invertible matrix is zero. This readily follows from the fact that the determinant can be viewed as a multivariate polynomial function of the matrix entries.
			For a matrix to be non-invertible, its parameters would need to satisfy the specific condition $\text{det}(A) = 0$. In the space of all possible real parameter values, this zero-determinant condition defines a lower-dimensional algebraic variety, which has Lebesgue measure zero in the full parameter space \cite{woude:03,blanchini2008set,harary1967graphs,RAMOS2022110229}. A similar statement holds for the discretization of continuous-time LTI systems \cite{tnse2020}. In another perspective, the full-rankness can be addressed via \textit{maximum-matching} in the bipartite representation of the social graph. It is known that having a maximum-matching spanning all nodes implies that the associated adjacency matrix is structurally full-rank (in generic sense) \cite[Supplementary Information]{Liu-nature}. In this perspective, the self-cycles in social graph associated with $a_{ii}\neq 0$ represent a maximum-matching covering all nodes $i$, which implies $A$ is {almost-surely} full-rank in generic sense. }. 
Note that the nonzero diagonal entry $a_{ii}\neq 0$ assures that each individual node can weight its own past state along with the states of its neighbors, which is the key for modelling complex social opinion dynamics where past states considerably influence the current behaviour.
It is important to emphasize that we do not require any other assumption on the structure of $A$ or stability of $A$. Specifically, Assumption~\ref{ass_fullRank} neither necessitates nor excludes the irreducibility of matrix $A$. This means that the requirement for $A$ to be irreducible is not imposed because some individuals can act solely as followers or influencers. Consequently, these individuals may not transmit or receive data to or from others, respectively. Besides, the social system matrix may be  (potentially) unstable with $\rho({A})>1$.

\textbf{Multi-Agent Observations:}
	Social networks are often monitored by \emph{agents} that may represent several entities (e.g., service providers or social network managers) serving as information gatherers by taking local observations. Specifically, agent $i\in\{1,\ldots,N\}$ performs the following measurement 
	\begin{align} \label{eq_H_i}
		\mb{y}_k^i = H_i\mb{x}_k + \zeta_k^i,
	\end{align}
	where~$\mb{y}^i_k\in\mathbb{R}^{p_i}$ and $H_i\in\mathbb{R}^{p_i\times n}$ are respectively the measurements performed by agent~$i$ at time~$k$ and local observation matrix describing $p_i\in\mathbb{N}$ linear combination of collected data, and $\boldsymbol{\zeta}^i_k \in \mathbb{R}^{p_i}$ is measurement noise due to rounding errors, considered as standard Gaussian noise $\zeta^i_k\in\mathbb{R}^{p_i}$.

	\textbf{Communication Network:}
	The agents seek to reconstruct the social state of the network by exchanging information among other agents through the (multi-agent) network $\mathcal{G}_W=(\mc{V}_W,\mc{E}_W;W)$,
	where:
	\begin{itemize}
		\item $\mathcal{V}_W = \{1, 2, \ldots, N\}$ represents the agent set
		\item $\mathcal{E}_W \subseteq \mathcal{V}_W \times \mathcal{V}_W$ denotes communication links
		\item $W = [w_{ij}] $ is the data-fusion weight matrix with the weight $w_{ji}$ associated with the link $(i,j)$. Particularly, matrix $W$ is row-stochastic implying that $\sum_{j=1}^N w_{ij}=1$.
\end{itemize}

Due to the exchange occurring through different communication channels, in practice,  the dynamics of the \textit{observed} states by the multi-agent network becomes 
\begin{align}\label{eq_p_aug}
	\widehat{\underline{\mb{x}}}_{k+1} =&\overline{W} \widehat{\underline{\mb{x}}}_{k},
\end{align}
where $\widehat{\underline{\mb{x}}}_{k} := \left(\widehat{\mb{x}}_{k}; \widehat{\mb{x}}_{k-1}; \dots; \widehat{\mb{x}}_{k-\overline{\tau}} \right)$ is the augmented delayed observed state at agents, accounting for an upper bound delay of $\overline{\tau}\in\mathbb{N}$. The state $\widehat{\mb{x}}_{k}$ represents the observed (or estimated) social state at agents at time $k$.   
	Additionally,  $\overline{W} $ represents  the dynamics of the multi-agent network under latency defined as
\begin{align}\label{eq_aug_W}
		\overline{W}:= \left(
		\begin{array}{ccccccc}
			W_0 &  W_1  & W_2 &  \hdots & W_{\overline{\tau}-1}  &  W_{\overline{\tau}} \\
			I_{n} &   0_{n}  & 0_{n}  &\hdots  & 0_{n}& 0_{n}\\
			0_{n} & I_{n} & 0_{n}  &   \hdots  & 0_{n} & 0_{n}  \\
			0_{n} &  0_{n} & I_{n} &  \hdots  & 0_{n} & 0_{n}  \\
			\vdots & \vdots & \vdots & \ddots & \vdots & \vdots \\
			0_{n} & 0_{n} & 0_{n} &  \hdots & I_{n} & 0_{n}
		\end{array}	
		\right),
	\end{align}
	where $W_r$ matrices represent the \textit{data-sharing} matrix among the agents under time-delay $r$ with the entries defined as
	\begin{align} \nonumber
		W_r(i,j) = \left\{
		\begin{array}{ll}
			w_{ij}, & \text{if}~ \tau_{ij}=r,  \\
			0, & \text{otherwise},
		\end{array}\right.
	\end{align} 
	with $\tau_{ij}$ denoting the time-delay on the link $(j,i)$. Particularly we have,
	\begin{align} \label{eq_sumW}
	\sum_{r=0}^{\overline{\tau}} W_r = W.
\end{align}	
	The dynamics~\eqref{eq_p_aug} represents the general augmented consensus among the agents under time-delay. This dynamics can be divided into $N$ local dynamics denoted by $\widehat{\mb{x}}_{k} := \left(\widehat{\mb{x}}^1_{k}; \widehat{\mb{x}}^2_{k}; \dots; \widehat{\mb{x}}^N_{k} \right)$. For each agent $i$, the local augmented consensus is defined as 
	\begin{align} \nonumber
		\widehat{\mb{x}}^i_{k} =& w_{ii}\widehat{\mb{x}}^i_{k-1} \\\label{eq_cons_i} &+ \sum_{j\in\mathcal{N}_i} \sum_{r=0}^{\overline{\tau}} w_{ij}\widehat{\mb{x}}^j_{k-r} \mb{I}_{k-r,ij}(r),
	\end{align}
	with $\mb{I}_{k,ij}(r)$ as the indicator function of delay $r$ registered at time $k$ from agent $j$ to $i$ defined as
	\begin{equation} \label{eq_I}
		\mb{I}_{k,ij}(r) = \left\{
		\begin{array}{ll}
			1, & \text{if}~ \tau_{ij}=r,  \\
			0, & \text{otherwise}.
		\end{array}\right.
	\end{equation}
	with $\tau_{ij}$ as the delay on the link $(i,j)$.

The following assumption describes our time-delay model. 
	\begin{ass} \label{ass_tau}	
		\textbf{(Communication Delays):} For any communication link $(j,i)$ in the multi-agent network $\mc{G}_W$, the associated time-delay $\tau_{ij}$ satisfies:
		\begin{enumerate}
			\item \textbf{Boundedness}: $\tau_{ij} \in \{0, 1, 2, \ldots, \bar{\tau}\}$ where $\bar{\tau} \in \mathbb{N}$ is a finite upper bound;
			\item \textbf{Time-invariance}: $\tau_{ij}$ is constant over time;
			\item \textbf{Heterogeneity}: For distinct links $(j,i)$ and $(l,k)$, it is not required that $\tau_{ij} = \tau_{lk}$;
			\item \textbf{Observability}: The delay values are known to the respective agents through time-stamped communication protocols.
		\end{enumerate}
	\end{ass}	
{Assumption 2.2 on the \textit{time-invariance} of delays $\tau_{ij}$ for every $(j, i) \in \mc{E}_W$ is needed for row-stochasticity of the augmented matrix $\overline{W}$. This follows the fact that only one of the terms
$W_0(i, j)$, $W_1(i, j)$, $\dots$ ,$W_{\overline{\tau}}(i, j)$ is equal to $w_{ij}$ and the rest must be
equal to zero. Therefore, in the first block-row of $\overline{W}$, we have $\sum_{j=1}^{N(\overline{\tau}+1)}\overline{W}(i,j)=\sum_{j=1}^{N}w_{ij}$, and given a row-stochastic
matrix $W$, its augmented version $\overline{W}$ also remains row-stochastic under time-invariant delays, and Eq.~\eqref{eq_sumW} holds. This is a requirement for the proof of stability in Theorem~\ref{thm_tau*}. The case of \textit{time-varying} delays needs a different non-stochastic distributed estimation setup, which is one of our future research directions.}

\subsection{Formal Problem Statement}
	Given a social dynamics~\eqref{eq_sys1} under Assumption~\ref{ass_fullRank}, observation model~\eqref{eq_H_i}, communication network $\mathcal{G}_W$ with communication delays satisfying Assumption~\ref{ass_tau}, and potential agent failures, design a single time-scale distributed inference scheme $\{\widehat{\mathbf{x}}^i_{k|k}\}_{i=1}^N$ such that:
	\begin{enumerate}[(i)]
		\item be  Schur-stable for any heterogeneous delay bounded by $\tau^\ast$, i.e., the closed-loop error matrix introduced later (cf.~\eqref{eq_Ahat}) has spectral radius strictly less than one for all $\tau \le \tau^\ast$;
		\item preserve distributed observability and error stability after the failure of any set $F$ of at most $p$ observationally equivalent agents, i.e.,
		\[
		\operatorname{rank}\!\big(\mc{O}(W_{-F}\!\otimes\!A,\, D_{H-F})\big) = nN,
		\]
		where   \[
		D_H := \operatorname{blockdiag}\!\big[ H_1^\top H_1,\; H_2^\top H_2,\; \ldots,\; H_N^\top H_N \big],
		\]
		and $W_{-F}$ and $D_{H-F}$ are obtained from $W$ and $D_H$ by removing the rows/columns (and blocks) associated with the failed agents, and $\mc{O}(\cdot,\cdot)$ denotes the usual observability matrix of the pair.
	\end{enumerate}

\section{Main Results} \label{sec_mainResults}
In this section, we begin by introducing a distributed estimation algorithm (Section~\ref{sec:mainDistAlg}), which outlines the steps that various agents need to follow in order to accurately determine the social state. Following this, we introduce some auxiliary results (Section~\ref{sec:relAuxLem}) that later will allow us to provide a detailed examination of the conditions that are both necessary and sufficient to ensure the successful acquisition of the social state (Section~\ref{sec:proofStab}). Building upon this foundation, we then present a robust methodology designed to maintain the integrity of this process, even in the face of agent failures due to natural occurrences or external interventions (Section~\ref{sec:obsRec}).

\subsection{The Main Distributed Algorithm}\label{sec:mainDistAlg}
The distributed monitoring algorithm by each agent to retrieve the social state (at time $k$) in the presence of heterogeneous time delays is as follows:
	\begin{enumerate}[(i)]
		\item \textit{Prediction step:} Agent $i$ averages \emph{a-priori} estimates via consensus while accounting for possible delays:
		\begin{align} \nonumber
			\widehat{\mb{x}}^i_{k|k-1} =& w_{ii}A\widehat{\mb{x}}^i_{k-1|k-1} \\\label{eq_p} &+ \sum_{j\in\mathcal{N}_i} \sum_{r=0}^{\overline{\tau}} w_{ij}A^{r+1}\widehat{\mb{x}}^j_{k-r|k-r} \mb{I}_{k-r,ij}(r),
		\end{align}
		where $\widehat{\mb{x}}^i_{k|k-1}$ denotes the local inference of social state~$\mb{x}_k$. The inference is based on all the information at agent~$i$ and its neighboring agents $\mc{N}_i =\{j~|~(j,i)\in\mc{E}_W\}$ up to time~$k-1$. In other words, under latency, agent $i$ averages the delivered data from agents $j\in \mc{N}_i$ as they arrive.
		Additionally, $A\in\mathbb{R}^{n\times n}$ is the social dynamics in~\eqref{eq_sys1}, and $\mb{I}_{k,ij}(r)$ defined as
		in \eqref{eq_I}.
		Recalling that the delays are time-invariant, at every link $(j,i)$ for one and only one $0 \leq r \leq \overline{\tau}$, we have $\mb{I}_{k-r,ij}(r)\neq 0$. The dynamics~\eqref{eq_p} can be obtained from \eqref{eq_cons_i} by including the a-priori estimates, i.e., by substituting $\widehat{\mb{x}}^j_{k-r} = A^{r+1}\widehat{\mb{x}}^j_{k-r|k-r}$ and $\widehat{\mb{x}}^i_{k-1} = A\widehat{\mb{x}}^i_{k-1|k-1}$ in Eq.~\eqref{eq_cons_i}.
		\\
		For \textit{delay-free} multi-agent communication network, the a-priori estimate at agent $i$ simplifies to
		\begin{align} \label{eq_delayfree}
			\widehat{\mb{x}}^i_{k|k-1} = w_{ii}A\widehat{{\mb{x}}}^i_{k-1|k-1}  + \sum_{j\in\mathcal{N}_i}  w_{ij}A\widehat{\mb{x}}^j_{k-1|k-1},
		\end{align} 
		where $A\widehat{{x}}^i_{k-1|k-1}$ is the prediction of $\widehat{{x}}^i_{k-1|k-1}$ at agent $i$ and $A\widehat{{x}}^j_{k-1|k-1}$ denotes the state prediction of $\widehat{{x}}^j_{k-1|k-1}$ at neighbouring agent $j \in \mathcal{N}_i$. Using the consensus weight matrix $W$, the inferred states are then averaged in the neighborhood of every agent $i$ to get the a-priori estimate. 
		\item \textit{Observation update:} Agent $i$ performs a measurement update (or innovation step) based on their partial sensing. This step involves no data-sharing, and every agent $i$ updates its prediction in~\eqref{eq_p} via its own local (partial) observation $\mb{y}^i_k$ of the social system:
		\begin{align}
			\label{eq_m}
			\widehat{\mb{x}}^i_{k|k} =& \widehat{\mb{x}}^i_{k|k-1} + K^i H_i^\top \left(\mb{y}^i_k-H_i\widehat{\mb{x}}^i_{k|k-1}\right),
		\end{align}
		with $K^i$ as the gain matrix at agent $i$ whose design is fully described in the Appendix\footnote{It should be clarified that in the design of the gain matrix $K$ in Appendix, we only need matrix $W$ to check the Schur stability of the closed-loop error dynamics, and matrix $W$ has no other role in the design of $K$ and in the dynamics \eqref{eq_m}.}.  This observation-update step can be modified to address possibly rank-deficient systems (not necessarily social models). This can be performed by incorporating data-sharing on the measurements in the neighborhood of every agent for observability recovery.
\end{enumerate}

\begin{remark}
	The proposed model \eqref{eq_p}-\eqref{eq_m} is a single \mbox{time-scale} estimator with no inner consensus loop. In other words, there is only one step of data exchange and consensus update between two consecutive time steps $k$ and $k+1$ of the social dynamics \eqref{eq_sys1}. \hfill$\circ$
\end{remark}

\subsection{Related Auxiliary Lemmas}\label{sec:relAuxLem}
Next, before we analyze the error stability under possible latency, to establish a successful retrieval of the social state, we need to introduce the following two lemmas. The first lemma analyzes the eigenspectrum of the augmented structured matrices.
\begin{lem} \label{lem_polynom}
	Consider the following block matrix of size $nN$ with $N \times N$ block matrices $P_{i}, I_N, 0_N$,
	\begin{align}\label{eq_aug_A}
		\overline{P}_{n,i} := \left(
		\begin{array}{ccccc}
			0_N & \hdots & P_{i} & \hdots & 0_N \\
			I_N &   0_N & \hdots &\hdots & 0_N\\
			0_N &  I_N & \ddots   &  \hdots   & 0_N  \\
			\vdots & \vdots & \ddots &  \ddots & \vdots \\
			0_N & 0_N &  \hdots & I_N & 0_N
		\end{array}	
		\right).
	\end{align} \normalsize
	The first block-row includes matrix $P_i$ at the $i$th block, and the rest of the blocks are zero. Let denote the characteristic polynomials of $P_{i}$ and $\overline{P}_{n,i}$ respectively by $p(\lambda)$ and $q(\overline{\lambda})$. Then, $q(\overline{\lambda}) = \overline{\lambda}^{N(n-i)}p(\overline{\lambda}^i)$.
\end{lem}
\begin{proof}
	Consider the following structure to analyze the eigenspectrum of $\overline{P}_{n,i}$,
	\begin{align}
		\overline{\lambda} I_{nN} - \overline{P}_{n,i} :=
		\left(\begin{array}{cc}
			\mc{A}  & \mc{B}  \\
			\mc{C} & \mc{D}
		\end{array}\right).
	\end{align}
	The block-matrices $\mc{A}, \mc{B},
	\mc{C}, \mc{D}$ are respectively of size $N(i-1) \times N(i-1)$,   $N(i-1) \times N(n-i+1)$, $N(n-i+1) \times N(i-1)$, and $N(n-i+1) \times N(n-i+1)$ and are defined as
	
	\small \begin{align} \nonumber
		\mc{A} &= \left(
		\begin{array}{ccccc}
			\overline{\lambda} I_N & 0_N & \hdots & \hdots &  0_N \\
			-I_N &  \overline{\lambda} I_N  & \hdots &\hdots & 0_N\\
			0_N &  -I_N & \ddots   &  \hdots   & 0_N  \\
			\vdots & \vdots & \ddots &  \ddots & \vdots \\
			0_N & 0_N &  \hdots & -I_N & \overline{\lambda} I_N
		\end{array}	
		\right),
	\end{align}
	
	\begin{align}  \nonumber
		\mc{B} &= \left(
		\begin{array}{cccc}
			-P_i & 0_N & \hdots  & 0_N \\
			0_N &   0_N & \hdots & 0_N\\
			\vdots & \vdots & \vdots &   \vdots \\
			0_N & 0_N &  \hdots  & 0_N
		\end{array}	
		\right), \ \ \mc{C} =\left(
		\begin{array}{cccc}
			0_N &  \hdots & 0_N & -I_N   \\
			0_N & \hdots & 0_N & 0_N\\
			\vdots & \vdots & \vdots & \vdots \\
			0_N & 0_N &  \hdots &  0_N
		\end{array}
		\right), 
	\end{align}
	
	\begin{align} \nonumber \text{ and } \quad
		\mc{D} = \left(
		\begin{array}{ccccc}
			\overline{\lambda} I_N & 0_N & \hdots & \hdots & 0_N \\
			-I_N &  \overline{\lambda} I_N & \hdots &\hdots & 0_N\\
			0_N &  -I_N & \ddots   &  \hdots   & 0_N  \\
			\vdots & \vdots & \ddots &  \ddots & \vdots \\
			0_N & 0_N &  \hdots & -I_N & \overline{\lambda} I_N
		\end{array}
		\right).
	\end{align} \normalsize
	From the definition we have $p(\lambda)= \mbox{det}|\lambda I_N - P_i|$, and
	\begin{align} \label{eq_gfeh}
		q(\overline{\lambda}) = \mbox{det}|\overline{\lambda} I_{nN} - \overline{P}_{n,i}|= \mbox{det}|\mc{A}|\mbox{det}|\mc{D}-\mc{C}\mc{A}^{-1}\mc{B}|,
	\end{align}
	where
	
	\small\begin{align}
		\mc{A}^{-1} &= \left(
		\begin{array}{ccccc}
			\frac{I_N}{\overline{\lambda}} & 0_N & \hdots & \hdots & 0_N \\
			\frac{I_N}{\overline{\lambda}^2} &   \frac{I_N}{\overline{\lambda}} & \hdots &\hdots & 0_N\\
			\vdots & \vdots & \ddots &  \ddots & \vdots \\
			\frac{I_N}{\overline{\lambda}^{i-1}} & \frac{I_N}{\overline{\lambda}^{i-2}} &  \hdots & \frac{I_N}{\overline{\lambda}^{2}} & \frac{I_N}{\overline{\lambda}}
		\end{array}	
		\right), 
	\end{align}
	
	\begin{align}
		\mc{D}-\mc{C}\mc{A}^{-1}\mc{B} &=
		\left(
		\begin{array}{ccccc}
			\overline{\lambda} I_N-\frac{P_i}{\overline{\lambda}^{i-1}} & 0_N & \hdots & \hdots & 0_N \\
			-I_N &  \overline{\lambda} I_N & \hdots &\hdots & 0_N\\
			0_N &  -I_N & \ddots   &  \hdots   & 0_N  \\
			\vdots & \vdots & \ddots &  \ddots & \vdots \\
			0_N & 0_N &  \hdots & -I_N & \overline{\lambda} I_N
		\end{array}
		\right).
	\end{align} \normalsize
	Then,  $ \mbox{det}|\mc{D}-\mc{C}\mc{A}^{-1}\mc{B}|=\overline{\lambda}^{N(n-i)}\mbox{det}|\overline{\lambda} I_N-\frac{P_i}{\overline{\lambda}^{i-1}}|$ and $\mbox{det}|\mc{A}|= \overline{\lambda}^{N(i-1)}$. By putting these in~\eqref{eq_gfeh} we obtain 
	\begin{align}
		q(\overline{\lambda}) =  \overline{\lambda}^{N(n-i)}\mbox{det}|\overline{\lambda}^i I_N - P_i| = \overline{\lambda}^{N(n-i)} p(\overline{\lambda}^i).
	\end{align}
\end{proof}

For stability analysis in the presence of latency, we need the following lemma.

\begin{lem} \label{lem_eig_aug}
	Given any matrix $W$ satisfying $\rho(W)<1$ and the  augmented system matrix $\overline{W}$ defined in~\eqref{eq_aug_W}, we have $\rho(\overline{W})\leq \rho(W)^{\frac{1}{\overline{\tau}+1}}<1$.
\end{lem}
\begin{proof}
	Recall that, for eigenspectrum analysis, the characteristic polynomial of $\overline{W}$ follows from Lemma~\ref{lem_polynom}. Consider characteristic polynomials $p(\lambda)$ and $q(\overline{\lambda})$ respectively for $W$ and $\overline{W}$.
	First, assume $\tau_{ij} = \overline{\tau}$ over the entire network. Then, $q(\overline{\lambda}) = p(\overline{\lambda}^{\overline{\tau}+1})$ and $\rho(\overline{W}) = \rho(W)^{\frac{1}{\overline{\tau}+1}}<1$.
	Note that, given $\rho(W)<1$, $\rho(W)^{\frac{1}{{\overline{\tau}}+1}}$ is monotonically increasing on $\overline{\tau}$. Then, for $\tau \leq \overline{\tau}$, we have $\rho(\overline{W}) = \rho(W)^{\frac{1}{\tau+1}}< \rho(W)^{\frac{1}{\overline{\tau}+1}} <1$. Similar reasoning implies that for $\tau_{ij} =r<\overline{\tau}$ and $W_{r} = W$ the following holds: $\rho(\overline{W}) = \rho(W)^{\frac{1}{r+1}}< \rho(W)^{\frac{1}{\overline{\tau}+1}} <1$.
\end{proof}

\subsection{Proof of Stability}\label{sec:proofStab}

For stability analysis, define the distributed inference error at agent $i$ as $\mb{e}_{k}^i := \mb{x}^i_{k} - \widehat{\mb{x}}^i_{k|k}$ and define the collective error vector $\mb{e}_{k}$ at all agents as,
\begin{align}\nonumber
	\mb{e}_{k} = \left( \begin{array}{c}
		\mb{e}_{k}^1\\
		\vdots \\
		\mb{e}_{k}^N
	\end{array}\right).
\end{align}

The main result of this paper can be stated as follows.

\begin{thm} \label{thm_tau*}
	Given a Schur stable delay-free error dynamics\begin{align}\label{eq_err0}
		\mb{e}_{k} = (W\otimes A - KD_H(W\otimes A))\mb{e}_{k-1} +
		{\mb{\eta}}_k,
	\end{align}
	with stabilizing gain ${K}:=\mbox{blockdiag}[K^i,\ldots,K^N]$ and column vector $\eta_{k} = [\eta^1_k;\dots;\eta^N_k]$, where
	\begin{align}
		\label{eq_eta0}
		\eta^i_k &:= \nu_{k-1}-H_i^\top H_i\nu_{k-1} -  H_i^\top \zeta^i_{k}, \\
		&\text{ and } D_H := \left(
		\begin{array}{cccc}
			H_1^\top H_1\\
			&\ddots\\
			& &H_N^\top H_N\
		\end{array}
		\right).
	\end{align} The distributed inference scheme \eqref{eq_p}-\eqref{eq_m} is Schur stable under Assumption~\ref{ass_fullRank}-\ref{ass_tau} for any bounded delay $\overline{\tau} \leq \overline{\tau}^*$ with  $\overline{\tau}^*$ satisfying
	\begin{align}\label{eq_tau*0}
		\rho(W\otimes A^{\overline{\tau}^*+1} - K \overline{D}_H (W\otimes A^{\overline{\tau}^*+1}) )< 1.
	\end{align}
	with $\overline{D}_H= D_H^\top D_H$.  	
\end{thm}
\begin{proof}
	For proof analysis, we perform the error analysis of the updated inference scheme \eqref{eq_p}-\eqref{eq_m}. Implementing the augmented consensus update, define the augmented social inference vector
	$$\underline{\widehat{\mb{x}}}_{k|k-1} := \left(\widehat{\mb{x}}_{k|k-1}; \widehat{\mb{x}}_{k-1|k-2}; \dots; \widehat{\mb{x}}_{k-\overline{\tau}|k-\overline{\tau}-1} \right)$$
	and similarly for $\underline{\widehat{\mb{x}}}_{k|k}$. Then, the global version of inference model \eqref{eq_p}-\eqref{eq_m} can be written in compact form as,
	\begin{align}\label{eq_p_aug0}
		\underline{\widehat{\mb{x}}}_{k|k-1} =&\overline{WA} \underline{\widehat{\mb{x}}}_{k-1|k-1},
		\\\label{eq_m_aug}
		\underline{\widehat{\mb{x}}}_{k|k} =&  \underline{\widehat{\mb{x}}}_{k|k-1} + \iota^{\overline{\tau}+1}_1 \otimes K D_H^\top \left(\mb{y}_k-D_H\Theta^{Nn}_{1,\overline{\tau}}\underline{\widehat{\mb{x}}}_{k|k-1}\right),
	\end{align}
	with $\overline{WA}$ representing the augmented version of $W \otimes A$ subject to latency as,
	\begin{align}\label{eq_aug_WA}
		\overline{WA}= \left(
		\begin{array}{cccccc}
			W_0 \otimes A & W_1\otimes A  &  \hdots & W_{\overline{\tau}-1}\otimes A  & W_{\overline{\tau}} \otimes A \\
			I_{Nn} &   0_{Nn}  &\hdots  & 0_{Nn}& 0_{Nn}\\
			0_{Nn} & I_{Nn} &   \hdots  & 0_{Nn} & 0_{Nn}  \\
			\vdots & \vdots &  \ddots & \vdots & \vdots \\
			0_{Nn} & 0_{Nn} &  \hdots & I_{Nn} & 0_{Nn}
		\end{array}	
		\right),
	\end{align}
	and
	\begin{align}	\nonumber
		\Theta^{Nn}_{i,\overline{\tau}}= (\iota^{\overline{\tau}+1}_i \otimes I_{Nn})^\top
	\end{align}	
	as an $Nn \times (\overline{\tau}+1)Nn $ matrix with $\iota^{\overline{\tau}+1}_i$ defined as,
	\begin{align}	\nonumber
		\iota^{\overline{\tau}+1}_i =( \underbrace{\overbrace{0;\dots;0}^{i-1};1;0;\dots;0}_{\overline{\tau}+1}).
	\end{align}	
	Note that in the proposed inference scheme, the matrix $\overline{WA}$ is not needed, and this is only defined for proof analysis and to simplify the mathematical illustration.
	
	For notation simplicity, we define
	$$\underline{\mb{x}}_{k} := \left(1_{N} \otimes   \mb{x}_k; 1_{N} \otimes \mb{x}_{k-1}; \dots; 1_{N} \otimes \mb{x}_{k-\overline{\tau}} \right).$$
	The augmented error vector can be described by $\underline{\mb{e}}_{k}  = \underline{\mb{x}}_{k} - \underline{\widehat{\mb{x}}}_{k|k}$ representing the difference of exact social states and estimated social states under latency.
	The overall observation model is defined as
	\begin{align}\nonumber
		\left(
		\begin{array}{c}
			\mb{y}^1_{k}\\
			\vdots\\
			\mb{y}^N_{k}
		\end{array}
		\right) &=
		\left(
		\begin{array}{c}
			H_{1}\\
			\vdots\\
			H_{N}
		\end{array}
		\right)\left(
		\begin{array}{c}
			x^1_{k}\\
			\vdots\\
			x^n_{k}
		\end{array}
		\right)+
		\left(
		\begin{array}{c}
			\zeta^1_{k}\\
			\vdots\\
			\zeta^N_{k}
		\end{array}
		\right),\end{align}
	which can be re-written as	
	\begin{align} \label{eq_y}
		\mb{y}_k &= H\mb{x}_k + \zeta_k,
	\end{align}
	with~$\mb{y}_k\in\mathbb{R}^{p}$,~$H=[H_{ij}]\in\mathbb{R}^{p\times n}$ (with $p=p_1+\ldots+p_N$), and $N$ as the number of agents. Without loss of generality, we assume $p_i=1$ to simplify the notation used in deriving the results.
	
	Ergo, the multi-agent error dynamics can be obtained as follows:
		\begin{align}
			\underline{\mb{e}}_{k}  =&\underline{\mb{x}}_{k} - \Bigl(\underline{\widehat{\mb{x}}}_{k|k-1} + \iota^{\overline{\tau}+1}_1 \otimes K D_H^\top \mb{y}_k-D_H\Theta^{Nn}_{1,\overline{\tau}}\underline{\widehat{\mb{x}}}_{k|k-1}\Bigr) \nonumber
			\\ \nonumber
			=& \underline{\mb{x}}_{k}  - \overline{WA}\underline{\widehat{\mb{x}}}_{k-1|k-1} \\  &- \iota^{\overline{\tau}+1}_1 \otimes K D_H^\top \left(\mb{y}_k-D_H\Theta^{Nn}_{1,\overline{\tau}}\overline{WA} \underline{\widehat{\mb{x}}}_{k-1|k-1}\right). \label{eq_love}
		\end{align}
		Assuming time-invariant delays (Assumption~\ref{ass_tau}),  for every link $(j,i) \in \mc{E}_W$, only one of the terms $W_0(i,j)$, $W_1(i,j)$, $\dots$, $W_{\overline{\tau}}(i,j)$ is equal to $w_{ij}$ and the rest are zero. This follows Assumption~\ref{ass_tau} and the fact that at every time-instant $k$, only one delayed message from agent $i$ is received by the agent $j$ (due to time-invariant delays). 
		Therefore, for a given row-stochastic consensus matrix $W$, its augmented version $\overline{W}$ is also row-stochastic; thus, we can write
		\begin{align} \label{eq_WA}
			\underline{\mb{x}}_{k}   = \overline{WA} \underline{\mb{x}}_{k-1} +  \underline{\mb{\nu}}_{k},
		\end{align}
		where $\widetilde{\mb{\nu}}_{k} :=  1_{N} \otimes \mb{\nu}_{k} $, $\widetilde{\mb{\zeta}}_{k} := 1_{N} \otimes \mb{\zeta}_{k} $, and $\underline{\mb{\nu}}_{k} :=\iota^{\overline{\tau}+1}_i \otimes\widetilde{\mb{\nu}}_{k}$. 
		Therefore, replacing \eqref{eq_WA} into \eqref{eq_love}, we obtain
		\begin{align} \nonumber
			\underline{\mb{e}}_{k} =& \overline{WA} \underline{\mb{x}}_{k-1} + \underline{\mb{\nu}}_{k} -  \overline{WA}\underline{\widehat{\mb{x}}}_{k-1|k-1} \\ \nonumber
			&-  \iota^{\overline{\tau}+1}_1 \otimes K D_H^\top \Big(D_H(1_{N} \otimes\mb{x}_k)+  \widetilde{\mb{\zeta}}_{k} \\
			&-D_H\Theta^{Nn}_{1,\overline{\tau}}\overline{WA}\underline{\widehat{\mb{x}}}_{k-1|k-1}\Big). \nonumber
		\end{align}
		Substituting system equations \eqref{eq_sys1} and \eqref{eq_y} and simplifying the outcome, we get
		\begin{align} \nonumber
			\underline{\mb{e}}_{k} =& \overline{WA} \underline{\mb{e}}_{k-1}
			- \iota^{\overline{\tau}+1}_1 \otimes K \overline{D}_H \Bigl((1_{N} \otimes A\mb{x}_{k-1}) \\ \label{eq_e1}
			&- \Theta^{Nn}_{1,\overline{\tau}}\overline{WA}\underline{\widehat{\mb{x}}}_{k-1|k-1}\Bigr) + \underline{\mb{\eta}}_k
		\end{align}
		where
		\begin{align} \nonumber
			\underline{\mb{\eta}}_k &= \underline{\mb{\nu}}_{k} - \iota^{\overline{\tau}+1}_1 \otimes K \overline{D}_H\widetilde{\mb{\nu}}_{k} -\iota^{\overline{\tau}+1}_1\otimes  K {D}_H^\top\widetilde{\mb{\zeta}}_{k} \\
			&=  \iota^{\overline{\tau}+1}_1 \otimes (\widetilde{\mb{\nu}}_{k} - K \overline{D}_H\widetilde{\mb{\nu}}_{k} - K {D}_H^\top\widetilde{\mb{\zeta}}_{k}).
			\label{eq_eta}
		\end{align}
		as the global augmented noise parameter. Note that by setting $\overline{\tau}=0$ (and $\iota^{\overline{\tau}+1}_1 = 1$) the above equation is the same as Eq.~\eqref{eq_eta0} for $\eta_k$.
		Next, recall from Eqs.~\eqref{eq_aug_W} and \eqref{eq_aug_WA} and the definition of $\Theta^{Nn}_{1,\overline{\tau}}$ that 
			\begin{align}\nonumber
			\Theta^{Nn}_{1,\overline{\tau}}\overline{WA} = \left(
				\begin{array}{cccccc}
					W_0 &  W_1  &  \hdots & W_{\overline{\tau}-1}  &  W_{\overline{\tau}} \\
					W_0 &  W_1  &  \hdots & W_{\overline{\tau}-1}  &  W_{\overline{\tau}}\\
					\vdots & \vdots &   \hdots  & \vdots & \vdots\\
					W_0 &  W_1  &  \hdots & W_{\overline{\tau}-1}  &  W_{\overline{\tau}}
				\end{array}	
				\right),
			\end{align}
	Using the above along with Eq.~\eqref{eq_sumW} and the row-stochasticity of $W$ matrix (i.e., $\sum_{j=1}^N w_{ij}=1$), we have $\Theta^{Nn}_{1,\overline{\tau}}\overline{WA}\underline{\mb{x}}_{k-1}=1_{N} \otimes A\mb{x}_{k-1}$. Then, by substituting \eqref{eq_WA} in \eqref{eq_e1}, we have
		\begin{align} \nonumber
			\underline{\mb{e}}_{k} =& \overline{WA} \underline{\mb{e}}_{k-1}
			- \iota^{\overline{\tau}+1}_1 \otimes K \overline{D}_H \Bigl(\Theta^{Nn}_{1,\overline{\tau}}\overline{WA}\underline{\mb{x}}_{k-1} \\ \nonumber
			&- \Theta^{Nn}_{1,\overline{\tau}}\overline{WA}\underline{\widehat{\mb{x}}}_{k-1|k-1}\Bigr) + \underline{\mb{\eta}}_k \nonumber \\
			=& \overline{WA} \underline{\mb{e}}_{k-1}
			- \iota^{\overline{\tau}+1}_1 \otimes K \overline{D}_H \Theta^{Nn}_{1,\overline{\tau}}\overline{WA}\underline{\mb{e}}_{k-1}  + \underline{\mb{\eta}}_k \label{eq_proof_e}
		\end{align}
		Hence, the error dynamics can be described by
		\begin{align}
			\underline{\mb{e}}_{k} = \underline{\widehat{A}} \underline{\mb{e}}_{k-1}  + \underline{\mb{\eta}}_k, \label{eq_err1}
		\end{align}
		where
		\begin{align}
			\underline{\widehat{A}} := \overline{WA}
			- \iota^{\overline{\tau}+1}_1 \otimes K \overline{D}_H \Theta^{Nn}_{1,\overline{\tau}}\overline{WA}, \label{eq_Ahat}
		\end{align}
		and $\rho(\underline{\widehat{A}})<1$  when \eqref{eq_err1} is Schur stable.
	Note that Eq.~\eqref{eq_err0} as the error dynamics in the absence of delays follows from \eqref{eq_proof_e} by setting $\overline{\tau}=0$. 
	Given that $\rho(W\otimes A - KD_H(W\otimes A))<1$, the Schur stability in the presence of time-delay ($\overline{\tau}>0$) follows from Lemma~\ref{lem_eig_aug}. Recall Eq. \eqref{eq_aug_WA} and the definition of $\underline{\widehat{A}}$. Then, we have
	\begin{align} \label{eq_tau*}
		\rho(\underline{\widehat{A}})\leq \rho(W\otimes A^{\overline{\tau}+1} - K \overline{D}_H (W\otimes A^{\overline{\tau}+1}) )^{\frac{1}{\overline{\tau}+1}}<1.
	\end{align}
	This holds for all $\overline{\tau}\leq \overline{\tau}^*$ satisfying~\eqref{eq_tau*0}. 
\end{proof}

Note that the error dynamics \eqref{eq_err0} is associated with the delay-free case, for which the Schur stability implies that 
\begin{align} \label{eq_delayfree}
	\rho(W\otimes A - KD_H(W\otimes A)) <1.
\end{align} 
Theorem~\ref{thm_tau*} proves that in the presence of heterogeneous time-delays (at different links) with $\tau_{ij} < \overline{\tau}$ the Schur stability follows Eq.~\eqref{eq_tau*} with the error dynamics in the form \eqref{eq_proof_e}. This Schur stability holds for any $\overline{\tau} \leq \overline{\tau}^*$ with  $\overline{\tau}^*$ satisfying Eq.~\eqref{eq_tau*0}. This holds for any gain matrix $K$ satisfying Eq.~\eqref{eq_delayfree}. 

\begin{remark}
	For stable social dynamics with $\rho({A})<1$, we have  $\rho({A})^{\overline{\tau}+1}<\rho({A})$, and therefore, Schur stability of $\widehat{A}$ guarantees Schur stability of $\underline{\widehat{A}}$ for any value of $\overline{\tau}$. This implies that for stable social dynamics  $\rho({A})<1$ we can have $\overline{\tau}^* \rightarrow \infty$.
	\hfill $\circ$
\end{remark}

\begin{remark} \label{rem_phi}
   The upperbound $\Phi$ on the variance of the error $\mb{e}_k$ in the absence of delays and in steady state can be approximated by the results in \cite{ejc2,tnse_attack,ecc} as,
   \begin{align} \label{eq_phi}
   	   \Phi := \frac{\alpha_1 N\|Q\|_2+\alpha_2\|R\|_2}{N\beta}
   \end{align}
   with $\alpha_1 := \|I_{nN}-KD_H\|_2^2 $, $\alpha_2 := \|K\|_2^2$, and certain $\beta<1$ as a function of $\|\widehat{A}\|_2$.    	
\end{remark}

Building on Theorem~\ref{thm_tau*}, it becomes imperative to ensure that the delay-free error dynamics, as delineated in~\eqref{eq_err0}, exhibit Schur stability. Based on the Kalman stability theorem \cite{kalman:61}, the error dynamics \eqref{eq_err0} is stabilizable (in steady-state) if and only if the pair $(W \otimes A, D_H)$ is observable. This is known as \textit{distributed observability}, and the conditions required to be attained are described in the following lemma.
	\begin{lem} \label{lem_stablenodelay}
		The cumulative error dynamics~\eqref{eq_err0} is Schur stabilizable if the matrix $W$ is irreducible, i.e., the associated multi-agent network $\mc{G}_W$ is strongly connected.	
	\end{lem}
	\begin{proof}
		First, notice that the delay-free dynamics~\eqref{eq_err0} is Schur stabilizable if $(W \otimes A, D_H)$ is observable.   	
		One can model the combination of the social network and the multi-agent network as a network-of-networks model or a composite Kronecker-product network. Then, the observability of $W \otimes A$ follows the structure of the social system graph associated with $A$ and the multi-agent network associated with  $W$. This complex network is a Kronecker product network associated with $W \otimes A$ matrix. To analyze the observability of this Kronecker product network, we recall some results from \cite[Theorem~4]{tsipn}. Ref. \cite{tsipn}  proves the observability condition for the Kronecker product network of two digraphs $\mc{G}_1$ and $\mc{G}_2$ associated with the adjacency matrix in the form $\mc{A}_1 \otimes \mc{A}_2$ with $\mc{A}_2$ being generally full-rank. 
		In our notation, from~\cite[Theorem~4]{tsipn}, it readily follows that $(W \otimes A, D_H)$ is observable if $\mc{G}_W$ satisfies the following two properties:  it is \emph{(i)} strongly-connected and \emph{(ii)} self-damped, i.e., all individual agents possess inner dynamics that would correspond to self-loops in the graph.
		Since, by assumption, every agent assigns a non-zero weight to its own information while averaging the data  (i.e., $w_{ii} \neq 0$) in~\eqref{eq_p}, then the self-damped (or inner dynamics) condition (ii) is satisfied. In the context of observability, having matrix $A$ full-rank from Assumption~\ref{ass_fullRank} and matrix $W$ self-damped (with inner dynamics) implies that $W \otimes A$ is full-rank. Therefore, only (i) needs to be ensured to satisfy output-connectivity; hence, the strong connectivity of $\mc{G}_W$ is sufficient for output-connectivity and $(W \otimes A,D_H)$-observability follows.
	\end{proof}
	\begin{remark}
		Given that $(W\otimes A, D_H)$ is observable and the gain matrix $K$ is designed via the LMI in the Appendix, the error dynamics \eqref{eq_err0} (in the absence of time-delay) is Schur stable from Lemma~\ref{lem_stablenodelay}.\hfill $\circ$
\end{remark}
Recall that strong-connectivity of the multi-agent network $\mc{G}_W$ is the minimum connectivity requirement in the existing distributed filtering scenarios \cite{chen2012diffusion,bordignon2022partial,7463019,khan2014collaborative}. Without strong-connectivity the information of some agents is not accessible to some other agents. 

\begin{remark}
	The design of $K$ via the LMI in the Appendix is irrespective of the time-delays (following from theorem~\ref{thm_tau*}). In other words, the gain matrix $K$ is designed for the delay-free closed-loop model \eqref{eq_err0} and it also works for the dynamics \eqref{eq_err1} in the presence of time-delays $\overline{\tau} \leq \overline{\tau}^*$ satisfying \eqref{eq_tau*0}. This is significant, as it reduces the computational complexity of our methodology.  Specifically, this implies that to stabilize the delayed error dynamics \eqref{eq_err1} of order $nN(\overline{\tau}+1)$, we only need to apply the design method in the Appendix on matrix $\widehat{A}$ of order $nN$, to design the $nN$-by-$nN$ block-diagonal matrix $K$. This further implies that the exact knowledge of time delays is not needed to obtain $K$; hence, there is no need to update this in the proposed inference scheme \eqref{eq_p}-\eqref{eq_m} for different values of time delays.\hfill $\circ$
\end{remark}

To summarize this section, from Lemma~\ref{lem_stablenodelay}, we design the multi-agent network $\mc{G}_W$ to be strongly-connected. This ensures $(W\otimes A, D_H)$-observability (or distributed observability) and the design of the block-diagonal gain matrix $K$ in the Appendix ensures the error stability. Then, the convergence of the proposed filter \eqref{eq_p}-\eqref{eq_m} in the presence of delays $\overline{\tau} \leq \overline{\tau}^*$ (with $\overline{\tau}^*$ satisfying Eq.~\ref{eq_tau*0}) follows from Theorem~\ref{thm_tau*}. 

\begin{remark}
	In this paper, we make no assumption on the stability of the social dynamics $A$, and matrix $A$ can be unstable. This is clearly shown in the simulations in Section~\ref{sec_sim}.
\end{remark}

\begin{remark}
Bounded time-varying delays $\tau_{ij}(k)\le\bar{\tau}$ can be accommodated
within the present framework by exploiting the time-stamped protocol to buffer each
received packet to the uniform maximum age $\bar{\tau}$. Under this uniformization only
$W_{\bar{\tau}}=W$ is active, the augmented matrix $\overline{W}$ remains row-stochastic, and the
scheme reduces to the homogeneous constant-delay case of Lemma~2 with $\tau=\bar{\tau}$;
hence Theorem~1 and the gain design of the Appendix apply unchanged. The price is
conservativeness: the network operates permanently at the worst-case latency $\bar{\tau}$
and stores $O(\bar{\tau})$ samples per link, discarding the benefit of packets that arrive
early.
\end{remark}

One can further add `redundancy' to the network $\mc{G}_W$ to ensure preserving distributed observability in case of agent failure, as discussed next. 

\subsection{Observability Recovery after Agent Failure}\label{sec:obsRec}

Next, we address scenarios in which information from an agent becomes unavailable due to excessive delays, packet losses, faults, cyber-attacks, or other cybernetic issues. We refer to such instances as `failed agents'. This failure results in the loss of  $(W\otimes A, D_H)$ observability, causing the inference error dynamics, as described in~\eqref{eq_err0}, to become non-Schur stable.

Therefore, the primary objective hereafter is to restore $(W\otimes A, D_H)$ observability despite the presence of failed agents. We propose achieving this through the integration of \textit{observationally equivalent agents}. These agents are designed to replicate the functionality of the failed agents closely, thus compensating for the lost information and ensuring the system's overall observability and error stability. 

\begin{definition}
	\cite{tnse18} \label{def_obsrv}
	Consider a social  dynamics $A$ and two agents $i,j\in\{1,\ldots,N\}$ with observation matrices $H_i,H_j$ of social as described in \eqref{eq_sys1}-\eqref{eq_H_i}, respectively. Additionally, let $H_l$, $l\in\{i,j\}$ be a canonical vector in the form
	\begin{align}	\nonumber
		H_l =( \underbrace{\overbrace{0;\dots;0}^{l-1};1;0;\dots;0}_{n}).
	\end{align}	
	The two agents are \emph{observationally equivalent} if $\mbox{rank}(\mc{O}(A,H_i))=\mbox{rank}(\mc{O}(A,H_j))$, where  
	\begin{align}\nonumber
		\mc{O}(A,H_l) = \left(
		\begin{array}{c}
			H_l\\
			H_lA\\
			H_lA^2\\
			\vdots\\
			H_lA^{n-1}
		\end{array}
		\right)
	\end{align} 
	is the observability matrix associated with agent $l\in\{i,j\}$.\hfill $\circ$
\end{definition}

Note that the graph $\mc{G}_A=(\mc{V}_A,\mc{E}_A;A)$, where nodes $\mc{V}_A=\{1,\ldots,n\}$ denote the individuals, and the links $\mc{E}_A\subset\mc{V}_A \times \mc{V}_A$ denote relationships or social interactions such that $(i,j)\notin \mathcal{E}_A$ if and only if $a_{ji}= 0$. Furthermore, $\mc{G}_A$ can be decomposed into strongly connected components (SCCs), i.e., the component or subgraph $\mc{S}_i=(\mc{V}_{S_i},\mc{E}_{S_i})$ in which there is a directed path with links in $\mc{E}_{S_i}$ from node $i$ to $j$ for any pair $i,j \in \mc{V}_{S_i}$. Besides, we can have a partial order between different SCCs
as \textit{parent} and \textit{child} \cite{godsil,camsap11}. 
The SCC $\mc{S}_i$ is called parent \cite{camsap11}  if it has no outgoing link to the nodes in $\mc{V}_A\backslash\mc{S}_i$, otherwise it is called a child SCC -- see illustration in Fig.~\ref{fig_parentchild}. A similar classification of non-root/root SCCs is given in \cite{liu-pnas}.

\begin{figure} 
	\centering
	\includegraphics[width=3.5in]{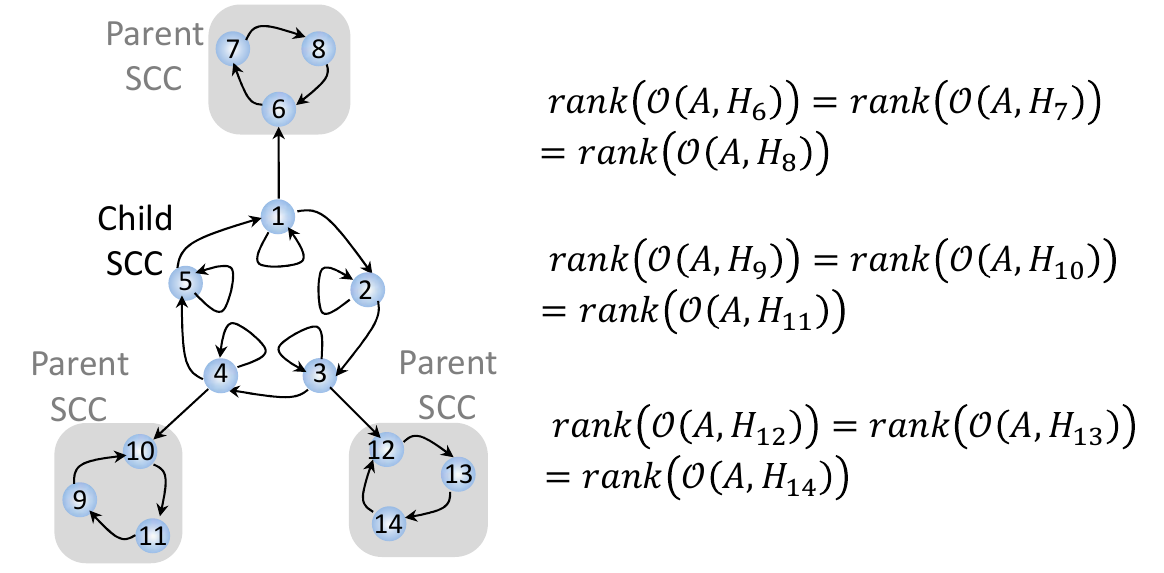}
	\caption{This figure shows a sample academic social graph $\mc{G}_A$ (which is not strongly connected) and its parent-child SCC classification. The SCCs with no outgoing links (to other SCCs) are called the parent, and the other remaining SCCs are called the child. The outputs of state nodes in the same parent SCC are observationally equivalent, i.e., the rank of observability matrix given the pair of system matrix $A$ and output $H_i$ (of state node $i$) are equal for all the nodes in the same parent SCC.
	} \label{fig_parentchild}
\end{figure}

\begin{ass} \label{ass_scc}
	It is assumed that every parent SCC includes at least two state nodes, i.e., there is no SCC with a single node. This is because, in the context of social networks, it is unlikely that a single individual only provides input to some of the others but receives information from none.
\end{ass}

\begin{thm}\label{lem:obsRecov} 
	Consider a social network under Assumption~\ref{ass_scc}, $\mc{G}_A=(\mc{V}_A,\mc{E}_A;A)$ with set of nodes $\mc{V}_A$ and links $\mc{E}_A$  associated with the (full-rank) social dynamics $A$ and two observationally equivalent agents $\alpha,\beta\in\{1,\ldots,N\}$ with canonical observation vectors $H_\alpha$, $H_\beta$ measuring any two nodes in the same parent SCC $\mc{S}_i$ in $\mc{G}_A$, for  $i\in\{1,\ldots,p\}$, where $p\in\mathbb{N}$ is the number of parent SCCs. If there are at least two observationally equivalent agents for every parent SCC $\mc{S}_i$ and the multi-agent network $\mc{G}_{W_{-f}}$ remains strongly connected after an agent $f$ fails, then $(W_{-f} \otimes A, D_H)$ is observable under the failure of any one observationally equivalent agent. 
\end{thm}
\begin{proof}
	Let $H_\mathcal{J}$ with $\mathcal J \in\{1,\ldots,N\}$ and $\mathcal J \neq \emptyset$ be a $n\times |\mathcal J|$ observation matrix consisting of the concatenation of $|\mathcal J|$ canonical observation vectors of agent $\gamma$, with $\gamma\in\mathcal J$. Additionally, let $\mc{G}=(\mc{V}_A\cup\mc{V}_W ,\mc{E}_A\cup\mc{E}_W\cup\mc{E}_{H_{\mathcal{J}}})$ be the overall graph comprising the social and multi-agent network, where $\mc{E}_{H_{\mathcal{J}}}$ describe the edges $(v^A_l,v^W_l)\in\mc{V}_A\times\mc{V}_W$ for agent $l\in \mathcal J$.

	From Theorem~\ref{thm_tau*} and Lemma~\ref{lem_stablenodelay}, it follows that we need to guarantee that $(A,H_{\mathcal{J}})$ is observable, and $\mc{G}_W$ is strongly connected. Besides, invoking \cite[Theorem~1]{camsap11} and results in \cite{alexandru2017limited}, it follows that a  necessary
	and sufficient condition for $(W\otimes A, D_{H_{\mathcal J}})$-observability is that, in $\mc{G}$,  for every agent $i \in \mc{V}_W$, there must exist a directed path from any $x_j \in \mc{V}_A$. 
	
	Now, notice that by construction, the above conditions hold as long there is an edge from a node in each of the parent SCCs  $\mc{S}_i$ in $\mc{G}_A$, to a node in $\mc{V}_W$, since  $\mc{G}_W$ is strongly connected; in fact it readily follows that  $\mc{G}_W$ is the only parent SCC in $\mc{G}$ when $(A,H_{\mathcal{J}})$ is observable. Subsequently, having two observationally equivalent agents $\alpha$ and $\beta$ for each $\mc{S}_i$ in $\mc{G}_A$ guarantees that two edges exist from $\mc{S}_i$ to $\mc{G}_W$. Hence, if the measurement of an observationally equivalent state in $\mc{S}_i$ fails, there remains one edge from $\mc{S}_i$ to $\mc{G}_{W_{-f}}$, which upholds the required conditions for the observability of $(W_{-f}\otimes A, D_{H_{\mathcal J}})$.  
\end{proof}

\begin{remark} From the proof of Lemma~\ref{lem:obsRecov}, it readily follows that we can ensure arbitrary resilience with respect to $p$-failures of observationally equivalent agents as long as at least one agent persists for each of the parent SCCs $\mc{S}_i$ in $\mc{G}_A$, and the remaining agents form a strongly connected multi-agent network.   \hfill $\circ$ 
\end{remark}

\begin{remark}
	There exist scalable and computationally efficient algorithms to design $p$-node-connected networks $\mc{G}_W$ that preserve strong-connectivity after the failure of up to $p$ nodes/agents. This is referred to as network augmentation or survivable network design. Interested readers may refer to \cite{augment_book,vegh2010connectivity,frederickson1981approximation,eurasip} for information on such network design algorithms. 
	\hfill $\circ$
\end{remark}

Theorem~\ref{lem:obsRecov} is illustrated in Fig.~\ref{fig_parentchild}. Note that this strategy is entirely graph-theoretic, i.e., irrespective of the numerical values of system parameters.
 
\begin{remark}
The SCC-decomposition and parent-child classification of the social network for observability recovery can be done via the depth-first search (DFS) algorithm with \mbox{polynomial-order} complexity $\mc{O}(N^2)$ \cite{algorithm}. The polynomial-order complexity implies that the algorithm is computationally efficient and scalable for large-scale social network applications. There are also distributed methods to find the SCCs \cite{9904000} and algorithms to find the minimum number of links to ensure contained SCC \cite{sp}.
\end{remark}

\section{Simulations} \label{sec_sim}

We illustrate the results of the paper in two different settings: \emph{(i)} a pedagogical example (Section~\ref{sec:sim1}), and \emph{(ii)} an illustrative real-world network example (Section~\ref{sec:sim2}).

\subsection{Pedagogical Example}\label{sec:sim1}

Consider the network depicted in Fig.~\ref{fig_parentchild} as $\mc{G}_A$. This social system includes the states (belief, opinion) of $14$ individuals, where its social digraph includes $4$ SCCs ($3$ parent SCCs and $1$ child SCC). In every SCC, the individuals make a social friendship circle. The entries $a_{ij}$ of the social system matrix (i.e., the weights of the associated links) imply the weight factor every individual puts on the state of the other individuals in the social digraph. These weights are set randomly as follows. 

\tiny\begin{align}\nonumber
	\left(
	\begin{array}{cccccccccccccc}
		0.83&    0.98&   0   & 0 & 0 &0.56&  0 & 0& 0 & 0& 0 & 0 & 0 & 0 \\
		0  &  0.73&  0.95&  0 & 0&0  & 0   & 0 & 0  & 0 & 0 & 0  &0  & 0\\
		0&0 &0.60& 0.98&  0 & 0 & 0 & 0 &0 & 0 & 0& 0.46& 0 & 0\\
		0 & 0  & 0 & 0.61& 0.20& 0 & 0 & 0 & 0 &0.80&  0 &0 & 0 & 0\\
		0.11& 0 & 0  & 0 & 0.13& 0 & 0 & 0  & 0 & 0 & 0 & 0 & 0 & 0 \\
		0& 0 & 0& 0 & 0 & 0 & 0.1&  0 &  0&0  &  0 &  0& 0 & 0 \\
		0 & 0 & 0 & 0&  0 & 0& 0 &1&  0 & 0 & 0 & 0  & 0&  0\\
		0 &0 & 0 & 0&  0& 0.22&  0& 0&  0 &  0& 0 &  0 &  0 & 0\\
		0 & 0 & 0 & 0 &0 & 0& 0& 0 & 0 & 0.16&  0 &  0&  0 &  0\\
		0&0 & 0& 0 & 0 &0& 0& 0 & 0 & 0 & 0.31&  0 &  0  &  0\\
		0 & 0 &0 & 0 & 0& 0& 0& 0 &1.04& 0 & 0 & 0 & 0 & 0\\
		0 & 0& 0 & 0 & 0 & 0  & 0 & 0& 0& 0& 0&  0 & 0.87& 0\\
		0 & 0& 0& 0& 0&0 & 0 & 0& 0 & 0 & 0& 0 & 0& 0.41\\
		0 & 0 & 0 & 0 & 0 & 0 & 0&  0&  0&  0 &  0 &0.26& 0 &  0								
	\end{array}	
	\right),
\end{align}\normalsize
The weights are such that the social dynamics is potentially unstable with $\rho(A)=1.086$. Three agents track the state of three individuals $\{7,9,12\}$ (one state node in every parent SCC) and estimate the state of the other unobserved individuals by sharing data over a simple directed cycle over time.
The noise parameters are set as  $\nu_k^i \sim \mc{N}(0,0.05)$ and $\zeta_k^i \sim \mc{N}(0,0.05)$. The maximum delay bounds are considered as $\overline{\tau} \in\{ 0,3,7,15\}$. Using the LMI design in the appendix the gain matrix ${K}:=\mbox{blockdiag}[K^i]$ with $K^i$ the local gain at agent $i$ follows as:

	\tiny\begin{align}\nonumber
		K^1 =	\left(
		\begin{array}{cccccccccccccc}
			0&    0&   0   & 0 & 0 &0&   0.24& 0 & 0 & 0& 0 & 0 & 0 & 0 \\
			0&    0&   0   & 0 & 0 &0&  0 & 0& 0 & 0& 0 & 0 & 0 & 0 \\
			0&    0&   0   & 0 & 0 &0&   2.07& 0 &0 & 0& 0 & 0 & 0 & 0 \\
			0&    0&   0   & 0 & 0 &0&   0.005& 0 & 0 & 0& 0 & 0 & 0 & 0 \\
			0&    0&   0   & 0 & 0 &0&   0.019& 0 & 0 & 0& 0 & 0 & 0 & 0 \\
			0&    0&   0   & 0 & 0 &0&  0 & 0& 0 & 0& 0 & 0 & 0 & 0 \\
			0&    0&   0   & 0 & 0 &0&  0 & 0& 0 & 0& 0 & 0 & 0 & 0 \\
			0&    0&   0   & 0 & 0 &0&   1.48& 0 & 0 & 0& 0 & 0 & 0 & 0 \\
			0&    0&   0   & 0 & 0 &0&   0.91&  0 &0 & 0& 0 & 0 & 0 & 0 \\
			0&    0&   0   & 0 & 0 &0&   0.002& 0 & 0 & 0& 0 & 0 & 0 & 0 \\
			0&    0&   0   & 0 & 0 &0&   -0.006& 0 & 0 & 0& 0 & 0 & 0 & 0 \\
			0&    0&   0   & 0 & 0 &0&  0 & 0& 0 & 0& 0 & 0 & 0 & 0 \\
			0&    0&   0   & 0 & 0 &0&   -0.13& 0 & 0 & 0& 0 & 0 & 0 & 0 \\
			0&    0&   0   & 0 & 0 &0&   -0.25& 0 & 0 & 0& 0 & 0 & 0 & 0 
		\end{array}	
		\right),
	\end{align}\normalsize

	\tiny\begin{align}\nonumber
		K^2 =	\left(
		\begin{array}{cccccccccccccc}
			0&    0&   0   & 0 & 0 &0&   0& 0 & 0 & 0& 0 & 0 & 0 & 0 \\
			0&    0&   0   & 0 & 0 &0&  0 & 0& 2.32 & 0& 0 & 0 & 0 & 0 \\
			0&    0&   0   & 0 & 0 &0&   0& 0 &0 & 0& 0 & 0 & 0 & 0 \\
			0&    0&   0   & 0 & 0 &0&   0& 0 & 0.009 & 0& 0 & 0 & 0 & 0 \\
			0&    0&   0   & 0 & 0 &0&  0 & 0& 0.033 & 0& 0 & 0 & 0 & 0 \\
			0&    0&   0   & 0 & 0 &0&  0 & 0& 0.052 & 0& 0 & 0 & 0 & 0 \\
			0&    0&   0   & 0 & 0 &0&   0& 0 &0.045   & 0& 0 & 0 & 0 & 0 \\
			0&    0&   0   & 0 & 0 &0&   0&  0 &0 & 0& 0 & 0 & 0 & 0 \\
			0&    0&   0   & 0 & 0 &0&   0& 0 & 0 & 0& 0 & 0 & 0 & 0 \\
			0&    0&   0   & 0 & 0 &0&  0& 0 & 1.017 & 0& 0 & 0 & 0 & 0 \\
			0&    0&   0   & 0 & 0 &0&  0 & 0& -0.01 & 0& 0 & 0 & 0 & 0 \\
			0&    0&   0   & 0 & 0 &0&  0& 0 & 1.22 & 0& 0 & 0 & 0 & 0 \\
			0&    0&   0   & 0 & 0 &0&   0& 0 & 0 & 0& 0 & 0 & 0 & 0	\\
			0&    0&   0   & 0 & 0 &0&   0& 0 & 0.004 & 0& 0 & 0 & 0 & 0  
		\end{array}	
		\right),
	\end{align}\normalsize

	\tiny\begin{align}\nonumber
		K^3 =	\left(
		\begin{array}{cccccccccccccc}
			0&    0&   0   & 0 & 0 &0&  0 & 0& 0 & 0& 0 & 0 & 0 & 0 \\
			0&    0&   0   & 0 & 0 &0&  0 & 0& 0 & 0& 0 & 1.24 & 0 & 0 \\
			0&    0&   0   & 0 & 0 &0&  0 & 0& 0 & 0& 0 & 0 & 0 & 0 \\
			0&    0&   0   & 0 & 0 &0&  0 & 0& 0 & 0& 0 & 0.013 & 0 & 0 \\
			0&    0&   0   & 0 & 0 &0&  0 & 0& 0 & 0& 0 & -0.005 & 0 & 0 \\
			0&    0&   0   & 0 & 0 &0&  0 & 0& 0 & 0& 0 & 0 & 0 & 0 \\
			0&    0&   0   & 0 & 0 &0&  0 & 0& 0 & 0& 0 & 0.05 & 0 & 0 \\
			0&    0&   0   & 0 & 0 &0&  0 & 0& 0 & 0& 0 & 0 & 0 & 0 \\
			0&    0&   0   & 0 & 0 &0&  0 & 0& 0 & 0& 0 & 0 & 0 & 0 \\
			0&    0&   0   & 0 & 0 &0&  0 & 0& 0 & 0& 0 & 0.083 & 0 & 0 \\
			0&    0&   0   & 0 & 0 &0&  0 & 0& 0 & 0& 0 & -0.002 & 0 & 0 \\
			0&    0&   0   & 0 & 0 &0&  0 & 0& 0 & 0& 0 & 1.036 & 0 & 0 \\
			0&    0&   0   & 0 & 0 &0&  0 & 0& 0 & 0& 0 & 0 & 0 & 0 \\
			0&    0&   0   & 0 & 0 &0&  0 & 0& 0 & 0& 0 & 0 & 0 & 0 
		\end{array}	
		\right),
	\end{align}\normalsize

For the designed inference scheme, the closed-loop error dynamics satisfies the spectral properties given in Table~\ref{tab_rho2}. Following Theorem~\ref{thm_tau*}, it readily follows that~\eqref{eq_tau*0} holds, which implies steady-state stable error dynamics. To verify this, the mean-squared estimation error is shown in Fig.~\ref{fig_error1} for a directed cyclic multi-agent network. The simulation is performed and averaged for $20$ independent Monte-Carlo trials. Following Remark~\ref{rem_phi} and Eq.~\eqref{eq_phi}, the upperbound on the error variance in steady-state can be approximated as $0.113$.

\begin{table}[hbpt!]
	\centering
	\caption{The spectral radius of the closed-loop observer for the multi-agent network in Fig.~\ref{fig_parentchild} under different bounds on the time-delays.}
	\begin{tabular}{|l|c|c|c|c|}
		\hline
		$\overline{\tau}$ & 0 & 3 & 7 & 15  \\
		\hline
		$\rho(\underline{\widehat{A}})$ & 0.871 & 0.970 & 0.982 & 0.990 \\	\hline 	\hline
	\end{tabular}
	\label{tab_rho2}
\end{table}

\begin{figure} 
	\centering
	\includegraphics[width=3in]{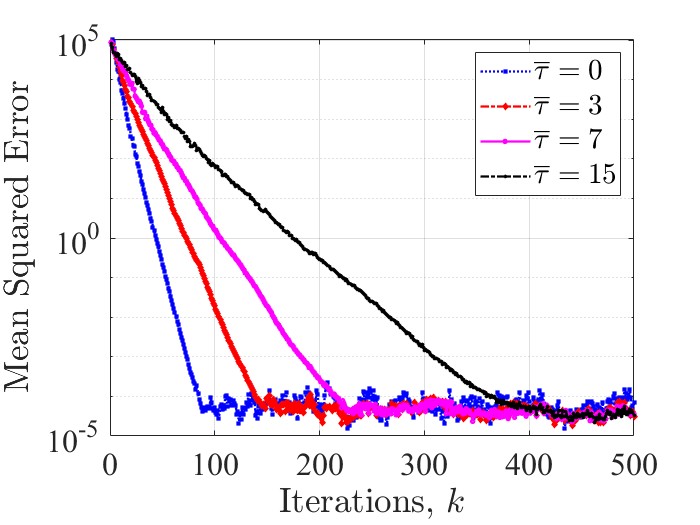}
	\caption{The time-evolution of the mean-squared error for distributed inference of the social network example in Fig.~\ref{fig_parentchild}. The simulation is performed and averaged for $20$ independent Monte-Carlo trials. 
	} \label{fig_error1}
\end{figure}

Next, we consider information loss at two agents with failed observations of states $\{7,9\}$. Following Definition~\ref{def_obsrv} and Theorem~\ref{lem:obsRecov}, we substitute these two by observationally equivalent agents observing states $\{8,10\}$ to recover the loss of distributed observability. For this case, the results associated with the closed-loop observer are given in Table~\ref{tab_rho3}. The mean-squared estimation error is shown in Fig.~\ref{fig_error2}. The simulation is performed and averaged for $20$ independent Monte-Carlo trials. Following Remark~\ref{rem_phi} and Eq.~\eqref{eq_phi}, the upperbound on the error variance in steady-state can be approximated as $0.144$. It readily follows from Theorem~\ref{thm_tau*}, and it is clear from Fig.~\ref{fig_error2} that the estimation error is Schur stable, implying that the distributed observability is recovered by the new setup (even in the presence of time-delays).

\begin{table}[hbpt!]
	\centering
	\caption{The spectral radius of the closed-loop observer for the recovered multi-agent network under different bounds on the time-delays.}
	\begin{tabular}{|l|c|c|c|c|}
		\hline
		$\overline{\tau}$ & 0 & 3 & 7 & 15  \\
		\hline
		$\rho(\underline{\widehat{A}})$ & 0.851 & 0.960 & 0.980 & 0.989 \\	\hline 	\hline
	\end{tabular}
	\label{tab_rho3}
\end{table}

\begin{figure} 
	\centering
	\includegraphics[width=3in]{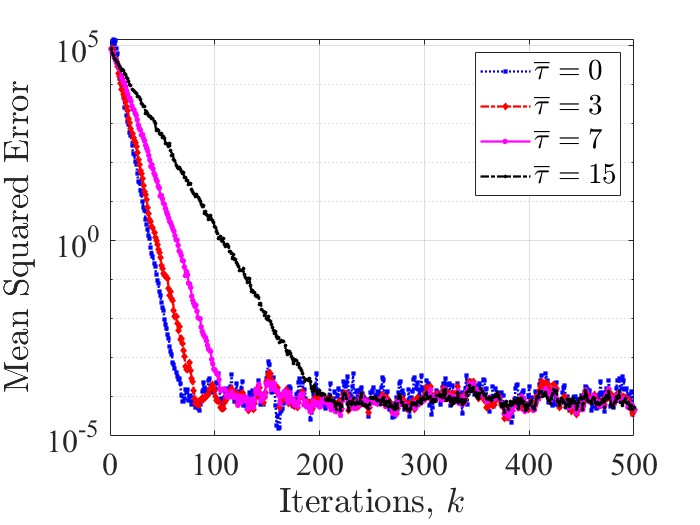}
	\caption{The time-evolution of the distributed inference mean-squared-error for the recovered multi-agent network. The simulation is performed and averaged for $20$ independent Monte-Carlo trials.
	} \label{fig_error2}
\end{figure}

\subsection{Sampson's Monastery Network}\label{sec:sim2}

Next, we consider Sampson's Monastery network given in \cite{sampson} for $\mc{G}_A$. This is a directed social network describing social relations among a group of $n=18$ men (novices) who were preparing to join a monastic order. The network includes $88$ social links with the weights set randomly such that the dynamics are unstable with $\rho(A)=1.1$. A directed cycle multi-agent network $\mc{G}_W$ of $N=4$ nodes is considered to track and infer this social network in a distributed way. The observation matrix $H_{\mathcal J}$ is defined by outputs from social states $\{1,5,12,18\}$. 

First, we assume that the multi-agent network is delay-free, and next, we extend to inference subject to heterogeneous time-delays bounded by different values of $\overline{\tau} = 5,10,20$. We randomly set the initial states of the agents. The noise variables are set as (Gaussian distribution) $\nu_k^i \sim \mc{N}(0,0.1)$ and $\zeta_k^i \sim \mc{N}(0,0.1)$. The block-diagonal feedback gain $K$ is designed via the methodology in the Appendix. Using this $K$, the closed-loop (delay-free) matrix satisfies Schur stable error dynamics as stated in Table~\ref{tab_rho1}, which provides evidence of the result in Theorem~\ref{thm_tau*}. This is better illustrated by the time-evolution of the mean-squared estimation error (at all agents) as shown in Fig.~\ref{fig_error_sampson}. The simulation is performed and averaged for $50$ independent Monte-Carlo trials. Following Remark~\ref{rem_phi} and Eq.~\eqref{eq_phi}, the upperbound on the error variance in steady-state can be approximated as $0.213$.
It is worth mentioning that if $\mc{G}_W$ is not strongly connected, the LMI optimization \eqref{eq_min} may have no solution, i.e., no block-diagonal gain matrix $K$ may exist to stabilize the error dynamics. 

\begin{table}[hbpt!]
	\centering
	\caption{The spectral radius of the closed-loop observer for the Sampson's network example under different bounds on the  time-delays}
	\begin{tabular}{|l|c|c|c|c|}
		\hline
		$\overline{\tau}$ & 0 & 5 & 10 & 20  \\
		\hline
		$\rho(\underline{\widehat{A}})$ & 0.461 & 0.879 & 0.932 & 0.964 \\		\hline 		\hline
	\end{tabular}
	\label{tab_rho1}
\end{table}

\begin{figure} 
	\centering
	\includegraphics[width=3in]{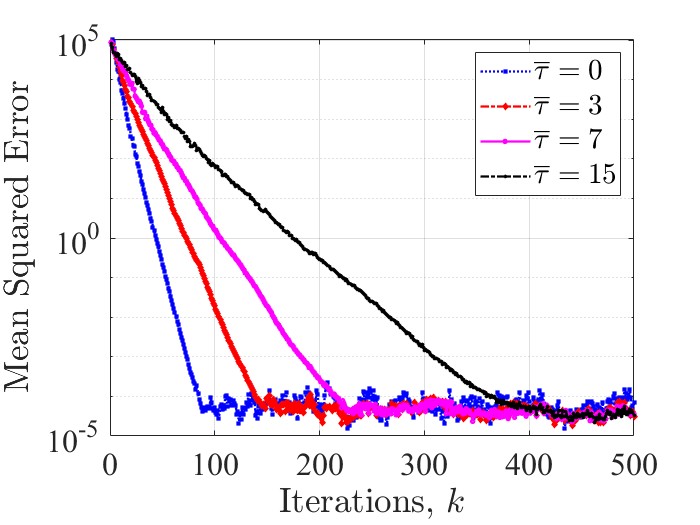}
	\caption{The time-evolution of the mean-squared error for distributed inference of Sampson's social network. The simulation is performed and averaged for $50$ independent Monte-Carlo trials.
	} \label{fig_error_sampson}
\end{figure}

\section{Concluding Remarks}  \label{sec_con}
This paper proposed a distributed scheme for inference on potentially unstable social dynamical networks using multi-agent systems, specifically designed to address challenges such as latency and agent failure. Our contributions include the development of a single-time-scale distributed inference model that effectively handles these challenges. Specifically, we proposed easily verifiable sufficient conditions that ensure the stability of the proposed scheme, even under the presence of heterogeneous bounded time-invariant delays. Additionally, we developed a computationally efficient recovery mechanism for agent failures that relies on employing graph-theoretic approaches to restore network observability by replacing failed agents (due to sensing failure or unbounded delays, i.e., packet drops) by implementing computationally efficient graph-theoretic methods to assign observationally equivalent agent counterparts.
The effectiveness and adaptability of our proposed model were demonstrated through both pedagogical examples and real-world network applications.

In future research, we aim to establish necessary and sufficient conditions for the stability of our proposed scheme. Developing distributed estimation protocols resilient to \textit{time-varying} delays is a future research direction. An extension that further considers model nonlinearities, e.g., polarization and epidemic spreading, is another future research direction. We also plan to explore various guarantees for the communication network within the multi-agent system to enhance its resilience. This will facilitate extending our methodology to distributed fault detection and isolation, as well as attack-detection scenarios. Additionally, we intend to develop practical stability guarantees that will enable the implementation of control schemes with assured performance. These advancements can significantly improve coordination practices to handle social unrest.

\section*{Acknowledgements}
This work is funded by Semnan University, research grant No. 226/1403/1403214.

\section*{Appendix: Block-Diagonal Feedback Gain Design}
The block-diagonal feedback gain matrix~$K$ in \eqref{eq_p}-\eqref{eq_m}
is the solution to the following \emph{linear matrix inequality}~(LMI):
\begin{equation}\label{LMI-9}
	\begin{aligned}
		& ~~ \left( \begin{array}{cc} Q&\widehat{A}^\top Q\\ Q\widehat{A}&Q\\ \end{array} \right) \succ 0\Rightarrow \left( \begin{array}{cc} Q&\widehat{A}^\top\\ \widehat{A}&R\\ \end{array} \right) \succ 0.
	\end{aligned}
\end{equation}
Recall that $\widehat{A} = W \otimes A -K D_H (W\otimes A)$ and ~$Q \succ 0$ (i.e., they must be positive definite). Note that the left equation above is nonlinear in~$K$; therefore, we adopt its linear equivalent proposed in \cite{pang:95,5717159} (the right equation in \eqref{LMI-9} with~$Q=R^{-1}$). Next, knowing that~$Q=R^{-1}$ is a non-convex constraint, it can be approximated by the following optimization problem \cite{rami:97}:
\begin{equation}
	\begin{aligned}
		\displaystyle
		\min
		~~ &  \mathbf{trace}(QR) \\
		\text{s.t.} ~~ & \left( \begin{array}{cc} Q&I\\ I&R\\ \end{array} \right) \succ 0,~Q,R \succ 0\\
		~~ & K\mbox{~is~block-diagonal}.\\
	\end{aligned}
\end{equation}
To summarize, knowing that distributed $(W\otimes A, D_H)$-observability holds, a feasible $K$ is the solution to the following optimization problem:
\begin{equation} \label{eq_min}
	\begin{aligned}
		\displaystyle
		\min
		~~ &  \mathbf{trace}(QR) \\
		\text{s.t.}  ~~& Q,R\succ 0,\\ ~~ & \left( \begin{array}{cc} Q&\widehat{A}^\top\\ \widehat{A}&R\\ \end{array} \right) \succ 0,~ \left( \begin{array}{cc} Q&I\\ I&R\\ \end{array} \right) \succ 0,\\
		~~ & K\mbox{~is~block-diagonal}.\\
	\end{aligned}
\end{equation}
Note that the solution to the second LMI is equivalent to~$Q=R^{-1}$, resulting in the optimal trace~$nN$. Replacing the term~$\mathbf{trace}(QR)$ with linear approximation~$\mathbf{trace}(\frac{R_0Q+Q_0R}{2})$ \cite{rami:97}, we can consider  Algorithm~\ref{alg_ac} to optimize \eqref{eq_min}. It is important to notice that \cite{rami:97} proves that~$\mathbf{trace}(R_tQ + Q_tR)$ is non-increasing and converges to~$2nN$. Therefore, the terminating criterion is set as reaching within~$2nN + \epsilon$ of the
trace objective. See~\cite{usman_cdc:11,pang:95,5717159,rami:97} for more details.
\begin{algorithm}
	\caption{The Iterative Algorithm to Design $K$}
	\begin{algorithmic}[1]
		\State \textbf{Input:}  $A,W,D_H$\;
		\State \textbf{Initialization:}  feasible~$X_0,Y_0,K$
		\State \textbf{While} {$\rho (\widehat{A}) \geq 1$} \textbf{do}
		{\State Minimize~$\mathbf{trace}(R_tQ + Q_tR)$ subject to constraints in~\eqref{eq_min} and find~$Q,R,K$\;
			\State		$Q_{t+1}=Q$\;
			\State		$R_{t+1}=R$\;
			\State		$t=t+1$\;}	
	\end{algorithmic}
	\label{alg_ac}
\end{algorithm}

The iterative convergence rate of the above algorithm depends on the size of the system and sparsity of $W$, i.e., for \textit{sparse} strongly-connected networks $\mc{G}_W$ it may take more iterations for Algorithm~\ref{alg_ac} to converge.

\bibliographystyle{elsarticle-num} 
\bibliography{bibliography}

\end{document}